\documentclass[aps,prx,reprint,groupedaddress,longbibliography,eprint,noshowkeys,citeautoscript]{revtex4-2}

\pdfoutput=1

\usepackage[T1]{fontenc}
\usepackage[english]{babel}
\usepackage{microtype}
\usepackage{amsmath}
\usepackage{amssymb}
\usepackage{amsthm}
\usepackage{mathtools}
\usepackage{mathrsfs}
\usepackage{csquotes}
\usepackage{enumitem}
\usepackage[table,dvipsnames]{xcolor}
\usepackage[colorlinks=true, ]{hyperref}
\usepackage{xurl}
\usepackage[capitalize]{cleveref}
\hypersetup{colorlinks=true, allcolors=blue!50!black, breaklinks=true}

\crefname{section}{Sec.}{Sects.}

\newcommand{\ce}{\coloneqq}
\newcommand{\wt}[1]{\widetilde{#1}}
\newcommand{\N}{\mathbb{N}}
\newcommand{\Z}{\mathbb{Z}}
\newcommand{\R}{\mathbb{R}}

\newcommand{\MF}{\mathcal{F}}
\newcommand{\MH}{\mathcal{H}}
\newcommand{\MN}{\mathcal{N}}
\newcommand{\MS}{\mathcal{S}}
\newcommand{\MFn}{\mathfrak{n}}
\newcommand{\ii}{\mathrm{i}}
\newcommand{\ee}{\mathrm{e}}
\newcommand{\id}{\mathrm{Id}}

\newcommand{\vecx}{\mathbf{x}}
\newcommand{\vecX}{\mathbf{X}}
\newcommand{\vecq}{\mathbf{q}}
\newcommand{\set}[2][]{#1\{{#2}#1\}}
\newcommand{\abs}[2][]{#1\vert{#2}#1\vert}
\newcommand{\norm}[2][]{#1\Vert{#2}#1\Vert}
\newcommand{\braket}[2][]{#1\langle{#2}#1\rangle}
\DeclareMathOperator{\dom}{dom}

\newcommand{\FUBaffiliation}{\affiliation{Freie Universität Berlin, Institute of Mathematics, Arnimallee 6, 14195 Berlin, Germany}}

\theoremstyle{plain}
\newtheorem{lemma}{Lemma}[section]
\newtheorem{corollary}[lemma]{Corollary}
\theoremstyle{remark}
\newtheorem{remark}[lemma]{Remark}

\begin{document}

\title{Chemical potential and variable number of particles control the quantum state: Quantum oscillators as a showcase}

\author{Benedikt M. \surname{Reible}}
\email{benedikt.reible@fu-berlin.de}

\author{Ana Djurdjevac}
\email{adjurdjevac@zedat.fu-berlin.de}

\author{Luigi \surname{Delle Site}}
\email{luigi.dellesite@fu-berlin.de}

\FUBaffiliation

\begin{abstract}
    Despite their simplicity, quantum harmonic oscillators are ubiquitous in the modeling of physical systems. They are able to capture universal properties that serve as reference for the more complex systems found in nature. In this spirit, we apply a model of a Hamiltonian for open quantum systems in equilibrium with a particle reservoir to ensembles of quantum oscillators. By treating (i) a dilute gas of vibrating particles and (ii) a chain of coupled oscillators as showcases, we demonstrate that the property of varying number of particles leads to a mandatory condition on the energy of the system. In particular, the chemical potential plays the role of a parameter of control that can externally manipulate the spectrum of a system and the corresponding accessible quantum states.
\end{abstract}


\maketitle

\section{introduction}

Ensembles of quantum harmonic oscillators and the ideal quantum gases are paradigmatic cases on which one can test the physical consistency and the utility of specific models for many-body systems. In this work, we will concretely investigate the model of an open quantum system which exchanges particles with a reservoir, that is, a system with a variable number of particles. This model has recently been discussed by some of us in Ref.~\cite{ana}, where an effective Hamiltonian for such systems was formally derived from the von Neumann equation governing the evolution of a density matrix.

The interest in quantum harmonic oscillators has its roots not only in the paradigmatic nature of the model, but also stems from the wide range of possible applications to various physical systems, which allow for the investigation of fundamental properties through manageable calculations. An ensemble of harmonic oscillators is essentially a universal and simple many-body system. For example, it is employed to detect general properties associated with heat flow, thus serving as a starting point for the construction of more complex models. Furthermore, in multiscale approaches, it is a precious tool for building hybrid continuum-particle models \cite{TTT16,TTT17,alex}, and it is also a very effective model for low-dimensional materials \cite{TTT18} and modern quantum technological devices \cite{TTT19}. At the same time, at a conceptual level, further progress is needed in the treatment of open quantum many-body systems with varying numbers of particles, as a deeper theoretical understanding of these systems is becoming increasingly important. The use of second quantization approaches with creation and annihilation of particles has been proposed as a further advancement in the design of quantum computers since there is a direct identification of the fermion occupation number basis and quantum memory of qubits \cite{secquantpap}. In general, a variable total particle number is a key problem in current research and in modern technological development; see, e.g., \cite{advl} and references therein. A relevant example is the use of such systems in addressing the challenge of controlling quantum coherence at the level of individual particles \cite{advphys}. As will be outlined below, in this paper we explore the consequences of a varying number of particles on systems of quantum harmonic oscillators, and we show that using the model of an open system leads to non-trivial physical consequences. In particular, we find that the accessible quantum states for bosons and states of interest for fermions depend on the chemical potential. In light of these findings and the technological perspective mentioned before, one may imagine to manipulate the chemical potential and the number of particles externally with the aim of obtaining a spectrum on demand, and/or to manipulate the environment externally to drive the system into targeted quantum states \cite{advphys}. Indeed, in modern quantum technology, ideas based on the control of the chemical potential are becoming increasingly popular \cite{science,nature}, and thus our contribution provides a basic formal rationale to otherwise intuitive technical procedures. Regarding the specific utility of the quantum harmonic oscillator in the context of systems with varying number of particles, this can be considered as a basic model with practical consequences for quantum technology. In fact, this prototype is the simplest and most efficient representation of a phononic system which can be directly employed in applied research; for example, fluctuations and noisy excitations of a phononic environment/reservoir can induce decoherence through changes in the qubit state with controlled consequences on stored quantum information. Furthermore, coherent phonons that propagate in a material carry quantum information, and they can lead to the distribution of quantum information on a chip \cite{decodeph, traveling}. The most illuminating example of utility, however, is the rather recent experimental design of quantum oscillators at room temperature, a physical result obtained with the aim of having an experimental prototype for all the basic applications suggested here \cite{expharm}. Further examples of direct and indirect utility of our model will be reported in the discussions in the following sections.

The paper is structured as follows. In \cref{sec:effectiveHamiltonian}, we will outline the basic ingredients of the model of an open quantum system that exchanges particles with a reservoir. Following this, we will review some aspects of the mathematical formalism necessary to treat ensembles of harmonic oscillators in \cref{sec:quantumOscillator}. Our results are contained in \cref{sec:energyConditions}, where we apply the model of \cref{sec:effectiveHamiltonian} to various systems of harmonic oscillators: first, we will consider a generic ensemble of independent oscillators, thus treating only vibrational properties, and we will look at the effect of the varying number of oscillators in the ensemble by deriving a necessary condition on the energy spectrum. Next, we will extend the treatment to a gas of independent quantum particles moving in space which are also characterized by internal vibrations; in essence, we treat the ensemble of oscillators of the first example and allow their centers of mass to have translational degrees of freedom. For example, one may think of a dilute gas of diatomic molecules or a basic model for fermion-phonon coupling, like a gas of electrons coupled to a phononic thermal bath \cite{phon-el} or spin systems, with interactions mediated by excitations of phonons, e.g., in a crystal \cite{mediated}. We will again demonstrate the effect of a variable number of particles on the energy of the system and derive a corresponding necessary condition for bosons that extends the condition of the open ideal Bose gas well-known from textbooks of statistical mechanics (see, e.g., \cite{huang} and \cite{mazenko}, and the remarks in \cref{app:idealBoseGas}). The interpretation of the results for the case of fermions requires further explanation, and a separate section is dedicated to a qualitative discussion of the meaning of such a condition for fermions. Finally, we will consider a one-dimensional chain of coupled harmonic oscillators, note its equivalence to the first example of independent oscillators, and thus extend our model also to this case by deriving once again a necessary condition on the energy spectrum of the system.

\section{Effective Hamiltonian for an open quantum system}\label{sec:effectiveHamiltonian}

In a previous paper of some of us \cite{ana}, an effective Hamiltonian for an open quantum system in contact with a particle reservoir was derived from first principles and found to be
\begin{equation}\label{eq:effectiveHamiltonian}
  H_n^\mathrm{eff} = H_n - \mu \MN \ ,
\end{equation}
with $H_n$ the Hamiltonian of $n \in \N$ instantaneous particles (that is, for a snapshot of the open system in which the particle number happens to be $n$). Furthermore, $\mu \in \R$ denotes the chemical potential, and $\MN$ is the number operator which counts the number of particles of the system (see \cref{sec:quantumOscillator}). Equation \eqref{eq:effectiveHamiltonian} was derived by tracing out the degrees of freedom of the environment in the von Neumann equation of a large system. The derivation is done for the situation of equilibrium or near-equilibrium, and the obtained Hamiltonian naturally leads to the expected solution of the von Neumann equation, namely the grand canonical density matrix of a system in equilibrium with its reservoir. In Ref. \cite{ana}, the full consistency of the effective Hamiltonian $H_n^\mathrm{eff}$ with models used in molecular simulation, (non-relativistic) quantum field theory, and the theory of superconductivity was also discussed.

Most importantly, for the considerations in the present paper, it must be noted that an effective Hamiltonian such as \eqref{eq:effectiveHamiltonian} leads to an effective energy spectrum which characterizes the system in light of its property of having a variable number of particles \cite{bogo,librobogo}. Indeed, using the operator $H_n^\mathrm{eff}$ instead of $H_n$ does not simply entail a trivial shift of the spectrum by a constant $\mu$, but rather it implies, above all, that the corresponding ground state requires a minimization with respect to the particle number $n$ as well, due to the statistical exchange of particles with a reservoir. In this sense, as the results of this paper will indicate, the operator $H_n^\mathrm{eff}$ is conceptually more fundamental than other quantities such as the particle number distribution which can be also used in this context. In fact, the latter can either be used as a cross-check of the results (in the case of bosons; see \cref{rem:conditionQ}) or as a complementary tool for the physical interpretation of the results obtained through the analysis of $H_n^\mathrm{eff}$ (in the case of fermions). In particular, for the case of coupled oscillators, the effective Hamiltonian approach is straightforward for the reduction of the model to independent oscillators, a result that would not be easy, if not impossible, to reach if one employs, e.g., the particle distribution.

\section{Quantum harmonic oscillators}\label{sec:quantumOscillator}

In the following, we will introduce the mathematical formalism which will be used in \cref{sec:energyConditions} below. First, we shall quickly review the algebraic diagonalization procedure for a single, three-dimensional harmonic oscillator. Then, we will discuss the mathematical setup for a system of $N \in \N$ independent, identical oscillators. Finally, we will introduce the formalism of second quantization for such systems.

\subsection{Single harmonic oscillator}\label{subsec:oneOscillator}

Consider a single harmonic oscillator of mass $m > 0$ and frequency $\omega > 0$ with dynamical variable $\vecx = (x^{(1)}, x^{(2)}, x^{(3)}) \in \R^3$. (For the physics of the quantum harmonic oscillator see, e.g., \cite{saku}; for mathematically rigorous treatments consult, for example, \cite{GustafsonSigal2020, Moretti2018, Teschl2014, Zeidler1995}.) The Hilbert space of this system is $\MH^\mathrm{osc} = L^2(\R^3)$, and the Hamiltonian takes the form \footnote{A suitable choice of domain for $H^\mathrm{osc}$ is $\dom(H^\mathrm{osc}) = \MS(\R^3)$, the space of Schwartz functions, whereby $H^\mathrm{osc}$ becomes essentially self-adjoint \cite{Moretti2018, Teschl2014, Zeidler1995}. Taking the closure of this operator yields the self-adjoint extension $\overline{H^\mathrm{osc}}$. The interested reader can find an explicit mathematical description of this operator in \cite{Moretti2018} and \cite{Zeidler1995}.\label{ftn:domain}}
\begin{equation}\label{eq:singleOscHamiltonian}
    H^\mathrm{osc} = \frac{p^2}{2m} + \frac{1}{2} \, m \omega^2 \abs{\vecx}^2 = - \frac{\hbar^2}{2m} \, \Delta_{\vecx} + \frac{1}{2} \, m \omega^2 \abs{\vecx}^2 \ .
\end{equation}
Note that the three components of the dynamical variable $\vecx$ are independent, hence one may factorize the Hilbert space as $\MH^\mathrm{osc} = L^2(\R) \otimes L^2(\R) \otimes L^2(\R)$, and the Hamiltonian \eqref{eq:singleOscHamiltonian} according to
\begin{align}\label{eq:decomposition1dOperators}
    \begin{split}
        H^\mathrm{osc} &= \sum_{\alpha=1}^{3} \left[- \frac{\hbar^2}{2m} \, \frac{\partial^2}{\partial (x^{(\alpha)})^2} + \frac{1}{2} \, m \omega^2 \bigl(x^{(\alpha)}\bigr)^2\right] \\
        &\equiv \sum_{\alpha=1}^{3} H_\alpha^\mathrm{osc} \ .
    \end{split}
\end{align}
We emphasize that $\sum_{\alpha=1}^{3} H_\alpha^\mathrm{osc}$ is a short-hand notation for the more precise
\begin{gather*}
    H_1^\mathrm{osc} \otimes \id_{L^2(\R)} \otimes \id_{L^2(\R)} + \id_{L^2(\R)} \otimes H_2^\mathrm{osc} \otimes \id_{L^2(\R)} \\
    + \id_{L^2(\R)} \otimes \id_{L^2(\R)} \otimes H_3^\mathrm{osc} \ ,
\end{gather*}
and that this operator is self-adjoint, given that the one-dimensional operators $H_\alpha^\mathrm{osc}$ are self-adjoint \cite{Schmüdgen2012, RS1} (see also Ref.~\cite{Note1}).

As is well-known, the one-dimensional Hamiltonian $H_\alpha^\mathrm{osc}$, $\alpha \in \set{1, 2, 3}$, can be diagonalized by introducing creation and annihilation operators \cite{saku, Moretti2018, GustafsonSigal2020}: define the symbols \footnote{A suitable domain for the creation and annihilation operators is again the space $\MS(\R)$ of Schwartz functions.}
\begin{align*}
     a_\alpha^\dagger &\ce \sqrt{\frac{m \omega}{2 \hbar}} \left(x^{(\alpha)} - \frac{\hbar}{m \omega} \, \frac{\partial}{\partial x^{(\alpha)}}\right) \ , \\
     a_\alpha &\ce \sqrt{\frac{m \omega}{2 \hbar}}\left(x^{(\alpha)} + \frac{\hbar}{m \omega} \, \frac{\partial}{\partial x^{(\alpha)}}\right) \ .
\end{align*}
Then it holds that $[a_\alpha, a_\beta^\dagger] = \delta_{\alpha \beta} \, \id$, and furthermore, there exists an orthonormal basis $(\psi_{q_\alpha})_{q_\alpha \in \N_0}$ of $L^2(\R)$, namely, the Hermite functions \footnote{The element $\psi_0 \in L^2(\R)$ is given for all $x \in \R$ by $\psi_0(x) \ce \pi^{-1/4} \, s^{-1/2} \, \ee^{- x^2 / (2 s^2)}$, where $s \ce \sqrt{\hbar / (m \omega)}$, and for $q_\alpha \in \N$ one defines $\psi_{q_\alpha} \ce (q_\alpha!)^{-1/2} \, (a_\alpha^\dagger)^{q_\alpha} \, \psi_0$; see Ref.~\cite{Moretti2018}.}, such that
\begin{gather*}
    a_\alpha \psi_0 = 0 \ , \\
    a_\alpha \psi_{q_\alpha} = \sqrt{q_\alpha} \, \psi_{q_\alpha - 1} \quad (q_\alpha \in \N) \ , \\
    a_\alpha^\dagger \psi_{q_\alpha} = \sqrt{q_\alpha + 1} \, \psi_{q_\alpha + 1} \quad (q_\alpha \in \N_0) \ .
\end{gather*}
Define also the number operator $\MFn_\alpha \ce a_\alpha^\dagger a_\alpha$. Then $(\psi_{q_\alpha})_{q_\alpha \in \N_0}$ are eigenvectors of $\MFn_\alpha$ with corresponding eigenvalues $q_\alpha$, that is,
\begin{equation*}
    \MFn_\alpha \psi_{q_\alpha} = q_\alpha \psi_{q_\alpha} \quad (q_\alpha \in \N_0) \ .
\end{equation*}
With the aid of the operator $\MFn_\alpha$, the one-dimensional Hamiltonian $H_\alpha^\mathrm{osc}$ can be rewritten as
\begin{equation}\label{eq:singleOscillatorNumberOperator}
    H_\alpha^\mathrm{osc} = \hbar \omega \left(\MFn_\alpha + \frac{1}{2} \, \id\right) \ ,
\end{equation}
and hence it is evident that $(\psi_{q_\alpha})_{q_\alpha \in \N_0}$ are eigenvectors of $H_\alpha^\mathrm{osc}$, corresponding to the eigenvalues $\set[\big]{\hbar \omega \left(q_\alpha + \frac{1}{2}\right) \, : \, q_\alpha \in \N_0}$. From the decomposition in \cref{eq:decomposition1dOperators,eq:singleOscillatorNumberOperator}, it follows that
\begin{equation*}
    H^\mathrm{osc} = \sum_{\alpha=1}^{3} \hbar \omega \left(\MFn_\alpha + \frac{1}{2} \, \id\right)
\end{equation*}
has the eigenvectors $\psi_{\vecq} = \bigotimes_{\alpha=1}^{3} \psi_{q_\alpha}$, where $\vecq = (q_1, q_2, q_3) \in \N_0^3$, and spectrum $\sigma(H^\mathrm{osc}) = \set[\big]{\sum_{\alpha=1}^{3} \hbar \omega \bigl(q_\alpha + \frac{1}{2}\bigr) \, : \, q_\alpha \in \N_0}$ \cite{GustafsonSigal2020, Teschl2014}.

\subsection{Extension to many oscillators}\label{subsec:manyOscillators}

Next, consider a system of $N \in \N$ independent, identical quantum harmonic oscillators described by dynamical variables $\vecx_i \in \R^3$, $1 \le i \le N$. According to the usual postulates of non-relativistic quantum mechanics \cite{Moretti2018}, the Hilbert space of this system is given by
\begin{equation*}
    \MH_N^\mathrm{osc} = \bigotimes\nolimits^{N} \MH^\mathrm{osc} = \bigotimes\nolimits^{N} L^2(\R^3) = \bigotimes\nolimits^{3N} L^2(\R) \ .
\end{equation*}
Since the oscillators are assumed to be uncoupled, the total Hamiltonian of the system takes the form
\begin{equation}\label{eq:totalHamiltonian}
    H_N^\mathrm{osc} = \sum_{i=1}^{N} H_i^\mathrm{osc} \ ,
\end{equation}
where each $H_i^\mathrm{osc}$ is given by \eqref{eq:singleOscHamiltonian}, and the sum of operators is understood in an analogous way as the one in \cref{eq:decomposition1dOperators}, i.e., $H_N^\mathrm{osc} = \sum_{i=1}^{N} \sum_{\alpha=1}^{3} H_{i, \alpha}^\mathrm{osc}$. Introducing creation and annihilation operators for each $H_{i, \alpha}^\mathrm{osc}$ separately, that is, using the representation \eqref{eq:singleOscillatorNumberOperator}, $H_N^\mathrm{osc}$ can be written as a sum of $3N$ one-dimensional operators (which shall be enumerated by a single index $\alpha \in \set{1, \dotsc, 3N}$ for better readability):
\begin{equation}\label{eq:totalHamiltonianNumberOp}
    H_N^\mathrm{osc} = \sum_{\alpha=1}^{3N} \hbar \omega \left(\MFn_\alpha + \frac{1}{2} \, \id\right) \ .
\end{equation}
In analogy to the remarks made at the end of \cref{subsec:oneOscillator}, one obtains that the eigenvectors $\Psi_N \in \MH_N^\mathrm{osc}$ of $H_N^\mathrm{osc}$ are given by tensor products of the functions $\psi_{q_\alpha}$, $\Psi_N = \bigotimes_{\alpha=1}^{3N} \psi_{q_\alpha}$, and thence they are specified by $3N$ quantum numbers $\set{q_\alpha \in \N_0 \, : \, 1 \le \alpha \le 3 N}$. Similarly, the spectrum of \eqref{eq:totalHamiltonianNumberOp} is found to be $\sigma(H_N^\mathrm{osc}) = \set[\big]{\sum_{\alpha=1}^{3N} \hbar \omega \bigl(q_\alpha + \frac{1}{2}\bigr) \, : \, q_\alpha \in \N_0}$ \cite{GustafsonSigal2020, Teschl2014}.

\subsection{Second quantization formalism}\label{subsec:secondQuantization}

Using the occupation number representation of the second quantization formalism (see, e.g., Refs.~\cite{Zeidler1995} and \cite{RS2,BR2,Coleman2015}), one may represent the states $\Psi_N$ from above in a different form. Let $q \in \N_0$ be the quantum number of the one-oscillator state $\psi_q$ (as introduced in \cref{subsec:oneOscillator}), and let $n_q \in \set{0, \dotsc, N}$ denote the number of times the state $\psi_q$ appears in the $N$-oscillator state $\Psi_N = \bigotimes_{\alpha=1}^{3N} \psi_{q_\alpha}$ from above. Since there are $N$ particles in total, we demand the following normalization condition:
\begin{equation}\label{eq:totalParticleNumber}
    \sum_{q=0}^{\infty} n_q = N \ .
\end{equation}
(Physically, this identity entails that adding together the number of oscillators in every possible quantum state $q \in \N_0$, one arrives at the total number of oscillators, $N$.)

Now, instead of describing the state $\Psi_N$ of the $N$-body system by specifying the $3N$ quantum numbers $\set{q_\alpha \in \N_0 \, : \, 1 \le \alpha \le 3 N}$ of the individual oscillators, as done in \cref{subsec:manyOscillators}, one may also prescribe it \cite{Strocchi2013} in terms of the (infinitely many) occupation numbers $\set{n_q \, : \, q \in \N_0}$; the corresponding representation of the wave function will be denoted by $\Phi_N \equiv \Phi_N(n_0, n_1, \dotsc, n_q, \dotsc)$. A suitable Hilbert space in which to consider the vectors $\Phi_N$, which also allows for the treatment of a variable number of oscillators, is the (symmetric / bosonic) Fock space \cite{RS1}:
\begin{equation*}
    \MF \ce \bigoplus_{N=0}^{\infty} \mathrm{Sym}(\MH_N^\mathrm{osc}) = \bigoplus_{N=0}^{\infty} \biggl(\bigotimes\nolimits^{3N} L^2(\R)\biggr) \ .
\end{equation*}
Elements of $\MF$ are sequences $(\Phi_N)_{N \in \N_0}$, where each $\Phi_N \in \MH_N^\mathrm{osc}$ is a symmetric $L^2$-function of $3N$ variables which is determined by a family $\set{n_q \, : \, q \in \N_0}$ of occupation numbers satisfying $\sum_{q \in \N_0} n_q = N$, such that $\sum_{N=0}^{\infty} \norm{\Phi_N}_N^2 < + \infty$, $\norm{\,\cdot\,}_N$ being the usual norm on the tensor product space $\MH_N^\mathrm{osc}$ \cite{Zeidler1995}.

In analogy to the single harmonic oscillator (\cref{subsec:oneOscillator}), one may introduce for every $q \in \N_0$ creation and annihilation operators $A_q^\dagger, A_q$ on the Fock space \footnote{A suitable domain for these operators is the set of all $(\Phi_M)_{M \in \N_0} \in \MF$ which satisfy $\sum_{M=0}^{\infty} M \norm{\Phi_M}_M^2 < + \infty$, see Ref.~\cite{BR2}.} by specifying their action on each tensor power $\MH_N^\mathrm{osc}$ as follows \cite{Strocchi2013}:
\begin{widetext}
    \begin{align*}
        A_q^\dagger \, \Phi_N(n_0, n_1, \dotsc, n_q, \dotsc) &= \sqrt{n_q + 1} \, \Phi_{N+1}(n_0, n_1, \dotsc, n_q + 1, \dotsc) \ ,\\
        A_q \, \Phi_N(n_0, n_1, \dotsc, n_q, \dotsc) &= \sqrt{n_q} \, \Phi_{N-1}(n_0, n_1, \dotsc, n_q - 1, \dotsc) \ .
    \end{align*}
\end{widetext}
In words: these operators act on the total system and create (respectively, annihilate) particles in a fixed single-oscillator mode $q$. One can also introduce a number operator $\MN_q \ce A_q^\dagger A_q$ \footnote{The domain of this operator is given by $\set[\big]{(\Phi_M)_{M \in \N_0} \in \MF \, : \, \sum_{M=0}^{\infty} M^2 \norm{\Phi_M}_M^2 < + \infty}$, see Refs.~\cite{Zeidler1995} and \cite{BR2}.} which acts on the states $\Phi_N$ according to
\begin{equation}\label{eq:eigenvaluesNumberOp}
    \MN_q \, \Phi_N(\dotsc, n_q, \dotsc) = n_q \Phi_N(\dotsc, n_q, \dotsc) \ .
\end{equation}
Using these operators, the Hamiltonian \eqref{eq:totalHamiltonian} can be represented in the occupation number basis as follows \cite{saku, GustafsonSigal2020}:
\begin{equation}\label{eq:secondQuantizedHamiltonian}
    H_N^\mathrm{osc} = \sum_{q=0}^{\infty} \hbar \omega \left(q + \frac{1}{2}\right) \MN_q \ .
\end{equation}

\section{Ensembles of quantum harmonic oscillators in a particle reservoir}\label{sec:energyConditions}

Consider an open system of quantum harmonic oscillators in contact with an infinite particle reservoir. The total particle number $n$ of the system is a finite but variable quantity described by the operator $\MN = \sum_{q=0}^{\infty} \MN_{q}$. (We follow the convention of the physics literature to denote a variable particle number by a lowercase $n$ and a fixed number of particles by an uppercase $N$.) The Fock space $\MF$ introduced above is the adequate Hilbert space for this system since, in this representation, the total number operator $\MN$ exists as a densely defined linear operator \cite{Strocchi2013}. Using \cref{eq:secondQuantizedHamiltonian}, the effective Hamiltonian $H_n^\mathrm{eff}$ from \cref{eq:effectiveHamiltonian} takes the form
\begin{equation}\label{eq:effectiveHamiltonianOscillators}
    H_n^\mathrm{eff} = \sum_{q=0}^{\infty} \hbar \omega \left(q + \frac{1}{2}\right) \MN_q - \mu \sum_{q=0}^{\infty} \MN_q \ ,
\end{equation}
where the condition \eqref{eq:totalParticleNumber} has to be imposed on the operators $\MN_q$ in order to fix the instantaneous number of particles (\enquote{closure condition}). In the following, we discuss physical conditions on the allowed energy spectrum of $H_n^\mathrm{eff}$ (bosons) and the allowed states of interest for an open system (fermions) for different situations.

\subsection{Condition on the purely vibrational energy spectrum}\label{subsec:conditionVibrationalSpectrum}

First, we discuss the spectrum of \eqref{eq:effectiveHamiltonianOscillators} directly to fix ideas and demonstrate the physical consequences; more realistic examples follow below. According to the discussion in \cref{subsec:secondQuantization}, see \cref{eq:eigenvaluesNumberOp} in particular, the energy eigenvalues of the Hamiltonian \eqref{eq:effectiveHamiltonianOscillators} in the occupation number states $\Phi_n$ are given by
\begin{equation}\label{eq:energySpectrumVibrations}
    E_n^\mathrm{eff} = \sum_{q=0}^{\infty} \hbar \omega \left(q + \frac{1}{2}\right) n_{q} - \mu \sum_{q=0}^{\infty} n_q \ ,
\end{equation}
with the single-oscillator energy level $q \in \N_0$ and $n_q \in \set{0, \dotsc, n}$ its occupation number. (Note that the occupation numbers follow a distribution; see \cref{eq:occupationNumberDist} in \cref{app:proofConvergence}. In particular, the range $\set{0, \dotsc, n}$ of $n_q$ mentioned before is just the integer part of the distribution for a specific $q$-value.) We emphasize that \eqref{eq:energySpectrumVibrations} is the consequence of two rigorous mathematical results: (i) the derivation of the effective open-system Hamiltonian from first principles in Ref.~\cite{ana}, and (ii) the well-known formalism of second quantization.

According to \cref{lem:convergenceOsc} proved in \cref{app:proofConvergence}, the effective energy can also be written as (see also \cref{cor:effectiveEnergies})
\begin{equation*}
    E_n^\mathrm{eff} = \sum_{q=0}^{\infty} \left[\hbar \omega \left(q + \frac{1}{2}\right) - \mu \right] n_q \ .
\end{equation*}
With this expression available, we can now formulate our two central physical assumptions for bosons. (i)  Since we are dealing with non-interacting bosons, each of which has positive energy that does not influence the energy of the others, we expect that $E_n^\mathrm{eff}$, considered as an \textit{effective} energy of the system, is positive:\footnote{From a mathematical standpoint, it is clear that for bosons the Hamiltonian \eqref{eq:secondQuantizedHamiltonian} of the oscillators is positive. The additional term $- \mu \MN$ in the effective Hamiltonian models a flux of particles between the system and the reservoir. Therefore, from a physical standpoint, one can expect that the energy remains positive because adding or removing particles cannot create negative energy states if particles are not interacting in any form.}
\begin{equation}\label{eq:vibrationalEnergyCondition}
    E_n^\mathrm{eff}= \sum_{q=0}^{\infty} \left[\frac{1}{2} \, \hbar \omega + q \hbar \omega - \mu\right] n_{q} > 0 \ .
\end{equation}
The situation is different for fermions as will be explained below. (ii) Furthermore, one can interpret the presence of the term $- \mu \MN$ in the effective Hamiltonian \eqref{eq:effectiveHamiltonianOscillators} as a shift of the vibrational frequency $\omega$, as was done, for example, by N. N. Bogoliubov in his model for fermionic superconductors \cite{bogo}. The rationale is the following: the term $- \mu \MN$ can be considered as a perturbation (mediating a particle flux) of the Hamiltonian $H_n^\mathrm{osc}$ that results in a change of the oscillation frequency of the system of harmonic oscillators. Looking at the expression \eqref{eq:vibrationalEnergyCondition} suggests to interpret the term in the square brackets as an effective oscillation frequency:
\begin{equation*}
    \hbar \omega_\mathrm{eff}(q) \ce \frac{1}{2} \, \hbar \omega + q \hbar \omega - \mu \quad (q \in \N_0) \ .
\end{equation*}
Then the effective energy is given by $E_n^\mathrm{eff} = \sum_{q=0}^{\infty} \hbar \omega_\mathrm{eff}(q) \, n_q$, which is a reasonable expression on physical grounds \cite{Coleman2015}. Now, since the effective oscillation frequency must be a positive quantity, we demand that $\hbar \omega_\mathrm{eff}(q) > 0$, i.e.,
\begin{equation*}
    \frac{1}{2} \, \hbar \omega + q \hbar \omega - \mu > 0 \ .
\end{equation*}
From this condition, one obtains a strict lower bound on the vibrational quantum number $q$, namely,
\begin{equation*}
    q > q_\mathrm{min}(\mu) \ce \frac{\mu}{\hbar \omega} - \frac{1}{2} \ .
\end{equation*}
This identity entails that there are quantum states which are not accessible by the system, depending on the chemical potential $\mu$. Furthermore, the effective energy, and hence the energy spectrum of the system, is modified according to
\begin{equation*}
    E_n^\mathrm{eff} = \sum_{q>q_\mathrm{min}(\mu)} \left[\frac{1}{2} \, \hbar \omega + q \hbar \omega - \mu\right] n_{q} \ .
\end{equation*}

As discussed above, the results of this section and their interpretation in terms of accessible states certainly hold for bosons. However, assumptions (i) and (ii) may not hold for fermions where even in the absence of a direct interaction, one still has the Pauli exclusion principle, which, as a matter of fact, countermands the particle independence by introducing correlations \cite{Coleman2015}. As a consequence, the addition or removal of particles cannot happen without the response of the other particles of the system. This means that bound states \footnote{In this context of open systems, we mean by ``bound states'' those that have an energy below the chemical potential, implying that they cannot migrate to the reservoir (see also \cref{subsec:fermions}).} may be possible due to the Pauli principle, and thus, the condition on $q$ necessary to guarantee a positive effective frequency $\omega_\mathrm{eff}(q)$ has a different interpretation compared to the case of bosons. This aspect will be discussed next.

\subsection{Interpretation of the condition for fermions}\label{subsec:fermions}

For an open system of independent fermionic oscillators (assumed spinless for simplicity), the average number of particles $\braket{n}$ is given by the Fermi-Dirac distribution \cite{huang,mazenko},
\begin{equation*}
    \braket{n} = \sum_{q=0}^{\infty} \frac{1}{\exp\bigl\{\beta \bigl[\hbar \omega \left(q + \frac{1}{2}\right) - \mu\bigr]\bigr\} + 1} \ .
\end{equation*}
Evidently, states satisfying $q < q_\mathrm{min}(\mu)$, where $q_\mathrm{min}(\mu) = \frac{\mu}{\hbar \omega} - \frac{1}{2}$ as above, will provide the most important contributions to $\braket{n}$ while states with $q > q_\mathrm{min}(\mu)$ represent only the tail of the distribution; this is a completely different situation compared to bosons, where the absence of bound states (in the sense specified in Ref.~\cite{Note7} and below) is reflected also in the particle distribution, cf. \cref{rem:conditionQ}. As a consequence of the indicated statistical preference of fermions towards states with $q < q_\mathrm{min}(\mu)$, the majority of occupied states would have negative effective energy, thus representing bound states. This is the expected behavior of the well-known fermionic oscillator model \cite{Woit2017}.

To clarify the meaning of the terminology ``bound state'' in such a situation, one can resort to the thermodynamic interpretation of the chemical potential at finite temperature and place it into the context of a system that statistically exchanges particles (incoming and outgoing) with a reservoir. The chemical potential is the average energy per particle that the system requires in order to allow a particle to be exchanged with the reservoir. This means that particles of the system occupying an energy state with energy smaller than $\mu$ do not have enough energy to pass, statistically, to the reservoir and, thus, are confined to the system. Instead, particles in energy states characterized by the condition $q > q_\mathrm{min}(\mu)$ are those that can be exchanged. The latter are actually the relevant particles for quantum devices based on exchange of information with a reservoir. A concrete experimental realization of this principle is given by low-dimensional spin systems in contact with microscopic electron currents induced by radical molecules at a surface \cite{radicals}.

The foregoing discussion about fermions for the case of pure oscillators can be extended straightforwardly to the cases treated in the next sections.

\subsection{Ideal quantum gas of vibrating particles}

Now, we consider a quantum gas of identical molecules contained in a finite cube $\Omega \ce [0, L]^3 \subset \R^3$ of side length $L > 0$, in which each molecule possesses, in addition to its vibrational degrees of freedom described by \eqref{eq:singleOscHamiltonian}, also translational degrees of freedom (e.g., a dilute gas of diatomic molecules).

Consider first a single molecule of mass $m$ whose translational motion is described by the dynamical variable $\vecX = (X^{(1)}, X^{(2)}, X^{(3)}) \in \R^3$. The corresponding dynamics is determined by the free Hamiltonian
\begin{align*}
    H^\mathrm{tra} &= \frac{P^2}{2m} = - \frac{\hbar^2}{2m} \, \Delta_\vecX = \sum_{\alpha=1}^{3} - \frac{\hbar^2}{2m} \frac{\partial^2}{\partial (X^{(\alpha)})^2} \\
    &\equiv \sum_{\alpha=1}^{3} H_\alpha^\mathrm{tra}
\end{align*}
in the Hilbert space $L^2(\Omega)$. Similarly to the Hamiltonian $H^\mathrm{osc}$ in \cref{subsec:oneOscillator}, the operator $H^\mathrm{tra}$ factorizes into three one-dimensional operators, corresponding to a factorization of the Hilbert space as $L^2(\Omega) = \bigotimes^3 L^2([0, L])$. We impose periodic boundary conditions by choosing as the domain for each one-dimensional operator $H_\alpha^\mathrm{tra}$ the subspace
\begin{equation*}
    \set[\big]{\psi \in H^2([0, L]) \ : \ \psi(0) = \psi(L) \ \text{and} \ \psi^\prime(0) = \psi^\prime(L)} \ ,
\end{equation*}
where $H^2([0, L])$ is the Sobolev space of second order \cite{Schmüdgen2012}. It holds that $H_\alpha^\mathrm{tra}$ is self-adjoint on this domain \cite{Schmüdgen2012}, and its spectrum is given by the eigenvalues \cite{saku, Teschl2014, Oliveira2009}
\begin{equation}\label{eq:translationEnergies}
    \varepsilon_{k_\alpha} = \frac{4 \pi^2 \hbar^2}{2 m L^2} \, k_\alpha^2 \ , \quad k_\alpha \in \Z \ .
\end{equation}
The corresponding eigenfunctions of $H_\alpha^\mathrm{tra}$ (plane waves), which form an orthonormal basis of $L^2([0, L])$, shall be denoted by $(\varphi_{k_\alpha})_{k_\alpha \in \Z}$.

Assume that the single molecule also possesses vibrational degrees of freedom, as discussed in \cref{subsec:oneOscillator}. The total Hamiltonian for the translational and vibrational motion acts in the Hilbert space
\begin{equation*}
    \MH^\mathrm{tot} = L^2(\Omega) \otimes L^2(\R^3) \cong \bigotimes\nolimits^3 \Bigl(L^2\bigl([0, L]\bigr) \otimes L^2(\R)\Bigr) \ ,
\end{equation*}
and it is formally given by
\begin{align*}
    H^\mathrm{tot} &= H^\mathrm{tra} \otimes \id_{L^2(\R^3)} + \id_{L^2(\Omega)} \otimes H^\mathrm{osc} \\
    &\equiv \sum_{\alpha=1}^{3} \bigl(H_\alpha^\mathrm{tra} + H_\alpha^\mathrm{osc}\bigr) \ .
\end{align*}
It holds that the operators $H_\alpha^\mathrm{tra} \equiv H_\alpha^\mathrm{tra} \otimes \id_{L^2(\R)}$ and $H_\alpha^\mathrm{osc} \equiv \id_{L^2([0, L])} \otimes H_\alpha^\mathrm{osc}$ strongly commute on the space $L^2([0, L]) \otimes L^2(\R)$ \cite{Schmüdgen2012}, thence we have
\begin{equation*}
    \bigl[H_\alpha^\mathrm{tra}, H_\alpha^\mathrm{osc}\bigr] = 0
\end{equation*}
on a suitable domain. Therefore, the spectrum of $H_\alpha^\mathrm{tra} + H_\alpha^\mathrm{osc}$ is given by the sum of the individual spectra, $\sigma(H_\alpha^\mathrm{tra} + H_\alpha^\mathrm{osc}) = \sigma(H_\alpha^\mathrm{tra}) + \sigma(H_\alpha^\mathrm{osc})$ \cite{Schmüdgen2012}, and its eigenvectors are $\varphi_{k_\alpha} \otimes \psi_{q_\alpha}$, $k_\alpha \in \Z$, $q_\alpha \in \N_0$.

The previous observations are readily generalized to a system of $N \in \N$ molecules with translational and vibrational degrees of freedom (see also the discussion in \cref{subsec:manyOscillators}): the total Hilbert space is given by
\begin{equation*}
    \MH_N^\mathrm{tot} = \bigotimes\nolimits^{3N} \Bigl(L^2\bigl([0, L]\bigr) \otimes L^2(\R)\Bigr) \ ,
\end{equation*}
and the total Hamiltonian can be written as
\begin{equation*}
    H_N^\mathrm{tot} = \sum_{\alpha=1}^{3 N} \bigl(H_\alpha^\mathrm{tra} + H_\alpha^\mathrm{osc}\bigr) \ .
\end{equation*}
The states of this system are specified by the $6 N$ quantum numbers $(k_\alpha, q_\alpha)$, $k_\alpha \in \Z$, $q_\alpha \in \N_0$, $1 \le \alpha \le 3 N$. Following the procedure outlined in \cref{subsec:secondQuantization}, instead of specifying an $N$-body state by these $6 N$ numbers, we may introduce for every $k \in \Z$ and $q \in \N_0$ number operators $\MN_{k, q} = A_{k, q}^\dagger A_{k, q}$ whose eigenvalues $n_{k, q} \in \set{0, \dotsc, N}$ count the number of molecules in each one-particle state $(k, q)$:
\begin{equation*}
    \MN_{k, q} \Phi_N(\dotsc, n_{k, q}, \dotsc) = n_{k, q} \Phi_N(\dotsc, n_{k, q}, \dotsc) \ .
\end{equation*}
As before in \cref{eq:totalParticleNumber}, we have to demand that $\sum_{k \in \Z} \sum_{q \in \N_0} n_{k, q} = N$. With this and $\varepsilon_k$ from \cref{eq:translationEnergies}, the total Hamiltonian $H_N^\mathrm{tot}$ can be represented as follows:
\begin{equation}\label{eq:totalHamiltonianTraOsc}
    H_N^\mathrm{tot} = \sum_{k \in \Z} \sum_{q \in \N_0} \left[\varepsilon_k + \hbar \omega \left(q + \frac{1}{2}\right)\right] \MN_{k, q} \ .
\end{equation}

Similarly to \cref{eq:effectiveHamiltonianOscillators}, the effective Hamiltonian \eqref{eq:effectiveHamiltonian} for an open system of molecules as characterized above, in which the particle number $n$ is now again a variable, is given by
\begin{align*}
    H_n^\mathrm{eff} &= \sum_{k \in \Z} \sum_{q \in \N_0} \left[\varepsilon_k + \hbar \omega \left(q + \frac{1}{2}\right)\right] \MN_{k, q} \\
    &\quad - \mu \sum_{k \in \Z} \sum_{q \in \N_0} \MN_{k, q} \ .
\end{align*}
In the occupation number state $\Phi_n(\dotsc, n_{k, q}, \dotsc)$, its eigenvalue is
\begin{align}\label{eq:energySpectrumTranslationsVibrations}
    \begin{split}
        E_n^\mathrm{eff} &= \sum_{k \in \Z} \sum_{q \in \N_0} \left[\varepsilon_k + \hbar \omega \left(q + \frac{1}{2}\right)\right] n_{k, q} \\
        &\quad - \mu \sum_{k \in \Z} \sum_{q \in \N_0} n_{k, q} \ .
    \end{split}
\end{align}
According to \cref{lem:convergenceTraOsc} and \cref{cor:effectiveEnergies} proved in \cref{app:proofConvergence}, this expression can also be written as
\begin{equation*}
    E_n^\mathrm{eff} = \sum_{k \in \Z} \sum_{q \in \N_0} \left[ \varepsilon_k + \frac{1}{2} \hbar \omega + \hbar \omega q - \mu \right] n_{k, q} \ .
\end{equation*}
With the same physical arguments as in \cref{subsec:conditionVibrationalSpectrum} regarding the expected positivity of the effective energy of a system of non-interacting particles, the mandatory condition for the energy spectrum is:
\begin{equation}\label{eq:energyIdealGasExist}
    E_n^\mathrm{eff} = \sum_{k \in \Z} \sum_{q \in \N_0} \left[ \varepsilon_k + \frac{1}{2} \hbar \omega + \hbar \omega q - \mu \right] n_{k, q} > 0 \ .
\end{equation}
As before, if one considers the presence of $\mu$ as a shift of the translational-vibrational frequency, thus interpreting the term in the square brackets as an effective frequency, then one obtains the stricter condition
\begin{equation}\label{eq:conditionQ}
    q > q_\mathrm{min}(\mu, k) \ce \frac{\mu}{\hbar \omega} - \frac{\varepsilon_k}{\hbar \omega} - \frac{1}{2} \quad \text{for all} \quad k \in \Z \ .
\end{equation}
Equations \eqref{eq:energyIdealGasExist} and \eqref{eq:conditionQ} extend the well-known condition for the existence of an ideal quantum gas of point-like bosonic particles obtained in the grand canonical ensemble, namely, $\varepsilon_k > \mu$ for all $k \in \Z$, see \cref{eq:existenceBoseGas} in \cref{app:idealBoseGas} and \cref{rem:conditionQ} in \cref{app:proofConvergence}.

\subsection{Linear chain of coupled harmonic oscillators}

Finally, we want to illustrate the use of the effective Hamiltonian \eqref{eq:effectiveHamiltonian} in the case of a linear chain of coupled, identical harmonic oscillators. (As reference, we use the formalism of Refs.~\cite{Coleman2015} and \cite{eisert}.) The Hamiltonian of $N \in \N$ coupled oscillators is given by \cite{eisert}
\begin{equation*}
    H_N^\mathrm{spr} = \sum_{j=1}^{N} \left[\frac{p_j^2}{2m} + \frac{1}{2} \, m \omega^2 x_j^2 + \frac{c}{2} \, m \omega^2 (x_j - x_{j+1})^2\right] \ ,
\end{equation*}
where $p_j$ is the momentum operator and $x_j$ the position operator of the $j$-th oscillator, and $c > 0$ is a coupling constant. Assuming periodic boundary conditions ($x_{j + N} = x_j$ and $p_{j + N} = p_j$) and performing a Fourier transformation of the coordinates and momenta,
\begin{align*}
    x_j &= \frac{1}{\sqrt{N}} \, \sum_{s=1}^{N} \ee^{2 \pi \ii s j / N} \, \wt{X}_s \ , \\
    p_j &= \frac{1}{\sqrt{N}} \, \sum_{s=1}^{N} \ee^{2 \pi \ii s j / N} \, \wt{P}_s \ ,
\end{align*}
the Hamiltonian can be written as
\begin{equation*}
    H_N^\mathrm{spr} = \sum_{s=1}^{N} \left[\frac{1}{2m} \, \wt{P}_s \wt{P}_{-s} + \frac{1}{2} \, m \omega_s^2 \wt{X}_s \, \wt{X}_{-s}\right] \ ,
\end{equation*}
where $\omega_s^2 = \omega^2 \left[1 + 4 c \sin^2\left(\frac{\pi s}{N}\right)\right]$. Introducing again (in analogy to \cref{subsec:oneOscillator}) creation and annihilation operators,
\begin{align*}
     a_s^\dagger &\ce \sqrt{\frac{m \omega_s}{2 \hbar}} \left(\wt{X}_{-s} - \frac{\ii}{m \omega_s} \, \wt{P}_{-s}\right) \ , \\
     a_s &\ce \sqrt{\frac{m \omega_s}{2 \hbar}}\left(\wt{X}_s + \frac{\ii}{m \omega_s} \, \wt{P}_s\right) \ ,
\end{align*}
and corresponding number operators $\MFn_s \ce a_s^\dagger a_s$, whose eigenvalues $q_s \in \N_0$ correspond to the energy level of the $s$-th oscillator, it follows that $H^\mathrm{spr}$ can be written in the familiar form
\begin{align*}
    H_N^\mathrm{spr} &= \sum_{s=1}^{N} \hbar \omega_s \left(a_s^\dagger a_s + \frac{1}{2} \, \id\right) \\
    &= \sum_{s=1}^{N} \hbar \omega_s \left(\MFn_s + \frac{1}{2} \, \id\right) \ .
\end{align*}

Evidently, $H^\mathrm{spr}$ has the same form as the Hamiltonian \eqref{eq:totalHamiltonianNumberOp} for an ensemble of independent harmonic oscillators treated in the previous sections, except that the frequency $\omega_s$ in the present case is different for every oscillator. One realizes, however, that one can still write it in a second quantized form, analogously to \eqref{eq:secondQuantizedHamiltonian}, as follows: let $q \in \N_0$ be arbitrary and define $\MN_q$ to be the number operator counting the number of times the one-particle state $q$ is occupied in the $N$-particle system, cf. \cref{eq:eigenvaluesNumberOp}. Furthermore, let $S_q \ce \set{s \, : \, 1 \le s \le N, \, q_s = q}$ be the set of $s$-values (that is, particle indices) for which the one-particle quantum number $q_s$ equals $q$. Then $H^\mathrm{spr}$ can be written equivalently as
\begin{equation*}
    H_N^\mathrm{spr} = \sum_{q=0}^{\infty} \, \Biggl[\,\sum_{s \in S_q} \hbar \omega_s \left(q + \frac{1}{2}\right)\Biggr] \, \MN_q \ .
\end{equation*}

With this expression, we obtain as an effective Hamiltonian for an open system with variable particle number $n$, in analogy to \cref{eq:effectiveHamiltonianOscillators}:
\begin{equation*}
    H_n^\mathrm{eff} = \sum_{q=0}^{\infty} \, \Biggl[\,\sum_{s \in S_q} \hbar \omega_s \left(q + \frac{1}{2}\right)\Biggr] \, \MN_q - \mu \sum_{q=0}^{\infty} \MN_q \ .
\end{equation*}
The eigenvalues of this operator in the occupation number states are given by
\begin{equation*}
    E_n^\mathrm{eff} = \sum_{q=0}^{\infty} \, \Biggl[\,\sum_{s \in S_q} \hbar \omega_s \left(q + \frac{1}{2}\right)\Biggr] \, n_q - \mu \sum_{q=0}^{\infty} n_q \ .
\end{equation*}
Therefore, as in the previous cases, the condition for the energy spectrum of an open linear chain of harmonic oscillators reads
\begin{equation*}
    E_n^\mathrm{eff} = \sum_{q=0}^{\infty} \, \Biggl[\,\sum_{s \in S_q} \hbar \omega_s \left(q + \frac{1}{2}\right) - \mu\Biggr] \, n_q > 0 \ .
\end{equation*}
Furthermore, if one interprets as before the presence of $\mu$ as a shift of the vibrational frequency of a single oscillator, thus treating the term inside the square brackets as an effective frequency, then one obtains the stricter condition
\begin{equation*}
    q > q_\mathrm{min}(\mu, n) \ce \frac{\mu}{\sum_{s \in S_q} \hbar \omega_s} - \frac{1}{2} \ .
\end{equation*}
Observe that the frequency $\omega_s$, and hence $q_\mathrm{min}$, depends on the particle number $n$ present in the open system. This feature implies the possibility of controlling $n$ by manipulating the external environment, and thus, it introduces a further parameter of control to determine  the accessible quantum states, in addition to the frequency shifting provided by $\mu$.

\section{Discussion and Conclusion}

We have applied the idea of an effective Hamiltonian for open quantum systems with varying number of particles to a prototype system of physics occurring ubiquitously in the modeling of many physical systems, the harmonic oscillator. We have shown that for an ideal gas of harmonic oscillators, there is a mandatory condition on the accessible quantum states and, hence, on the energy spectrum for bosons while representing the ``non-bound'' states for fermions, that is, the states of interest for the exchange of particles with a reservoir. This condition is directly linked to the chemical potential of the system, which, in this sense, can act as an external parameter of control of the spectrum and of the occupation number of particles in the available vibrational and translational states. The present result extends the well-known mandatory condition for an ideal gas of point-like bosonic particles; in the latter case, the energy spectrum is characterized only by the wave number associated with translation.

Furthermore, we have discussed that the open ensemble of independent harmonic oscillators can be extended to a chain of coupled oscillators with nearest-neighbor interactions, building on the formal reduction of the coupled Hamiltonian to the Hamiltonian of uncoupled oscillators, which is well-established in literature. The physics of such a system in the presence of a particle reservoir introduces additional features: by deriving an expression for the energy spectrum using the effective Hamiltonian of an open system and establishing, similarly to the case of an ensemble of non-interacting harmonic oscillators, a mandatory condition on the spectrum of the chain, we have found that the accessible part of the spectrum for bosons, respectively, the part relevant for information exchange for fermions, becomes explicitly $n$-dependent.

The utility of the current results for future applications lies in the fact that models of open systems with variable number of particles can cover a rather broad class of physical situations, certainly more than the model at fixed number of particles. An interesting example, mentioned in the introduction, concerns the idea of externally manipulating the chemical potential in order to control the spectral properties of systems used in modern quantum technology; in this context, our results provide an explicit, basic formal ground for such experimental protocols. From a larger perspective, one may imagine that these results can be useful for several other problems: for example, the open ensemble of harmonic oscillators with a variable number of particles may be used as a basic approach to the popular problem of treating entanglement in a quantum (Fermi) gas in an harmonic trap \cite{riccarda}. Moreover, the applications of this approach can be extended to the modeling and simulation of heat and mass transport in molecular systems \cite{roya1,roya2}, in low-dimensional materials \cite{TTT18}, and in computational modeling for hybrid particle-continuum models \cite{TTT17}.

From a conceptual point of view, the harmonic oscillator is a simple yet paradigmatic model for the formal derivation of explicit formulas that can describe universal generic properties. In this context, we mention as a concrete example the recent proposal \cite{carsten} for a possible definition of a ``non-equilibrium'' free energy, which was established in the framework of coupled harmonic oscillators and which gives rise to estimates of universal trends in heat transport in particle systems. In further investigations in the future, the results of the present paper can be combined with the results of Ref.~\cite{carsten}, and one can consider a broader class of physical systems and situations.

\begin{acknowledgments}
    This work was supported by the DFG Collaborative Research Center 1114 ``Scaling Cascades in Complex Systems'', project no. 235221301, projects  C01 (L.D.S. and B.M.R.) ``Adaptive Coupling of Scales in Molecular Dynamics and Beyond to Fluid Dynamics'' and C10 ``Numerical Analysis for Nonlinear SPDE Models of Particle Systems'' (A.Dj.). Further support by the DFG, project No. DFG-DE 1140/15, ``Mathematical Model and Computational Implementation of Open Quantum Systems for Molecular Simulations'' (L.D.S.) is acknowledged.
\end{acknowledgments}

\section*{Author declarations}

\subsection*{Conflict of Interest}

The authors have no conflicts to disclose.

\subsection*{Author Contributions}

\textbf{Benedikt M. Reible}: Conceptualization (lead); Formal Analysis (equal); Methodology (lead); Writing -- original draft (lead); Writing -- review \& editing (equal). \textbf{Ana Djurdjevac}: Formal Analysis (equal); Writing -- original draft (supporting); Writing -- review \& editing (equal). \textbf{Luigi Delle Site}: Conceptualization (lead); Formal Analysis (equal); Methodology (lead); Writing -- original draft (lead); Writing -- review \& editing (equal).

\section*{Data Availability}

Data sharing is not applicable to this article, as no new data
were created or analyzed in this study.

\appendix
\section{The ideal Bose gas}\label{app:idealBoseGas}

The grand canonical partition function for an ideal gas of bosons is given by \cite{huang,mazenko}
\begin{equation*}
    Z(\mu, V, T) = \prod_{k \in \Z} \frac{1}{1-e^{-\beta(\varepsilon_{k}-\mu)}} \ .
\end{equation*}
Here, $\mu \in \R$ is the chemical potential, $V > 0$ the volume, and $T > 0$ the temperature; furthermore, $k \in \Z$ labels the different one-particle states, and $\varepsilon_k = \hbar^2 k^2 / (2 m)$ denotes the corresponding energy. The average occupation number $\braket{n_k}$ of the state $k$ takes the form
\begin{equation*}
    \braket{n_k} = \frac{1}{\ee^{\beta (\varepsilon_k - \mu)} - 1} \ .
\end{equation*}
From this expression, one obtains a mandatory condition on the spectrum of the ideal Bose gas. Indeed, since based on physical grounds it must hold true that $\braket{n_k} \ge 0$ for all $k \in \Z$, it follows that
\begin{equation}\label{eq:existenceBoseGas}
    \varepsilon_k - \mu > 0
\end{equation}
for every state $k$. (The case of equality is excluded in order for $Z(\mu, T, V)$ and $\braket{n_k}$ to be well-defined.) In particular, this means that $\mu$ must be less than or equal to the ground state energy of the system.

\section{Proof of convergence of the effective energy}\label{app:proofConvergence}

In the following, it will be shown that the series appearing in the effective energies of the vibrational model \eqref{eq:energySpectrumVibrations} and the translational-vibrational model \eqref{eq:energySpectrumTranslationsVibrations} converge. Recall from \cref{eq:totalHamiltonianTraOsc} that the energy of a system of $n$ molecules with translational and vibrational degrees of freedom in the occupation number state $\Phi_n(\dotsc, n_{k, q}, \dotsc)$ is given by
\begin{equation}\label{eq:energyValuesTraOsc}
    E_n = \sum_{k \in \Z} \sum_{q \in \N_0} \left[\varepsilon_k + \hbar \omega \left(q + \frac{1}{2}\right)\right] n_{k, q} \ ,
\end{equation}
where $\varepsilon_k = 4 \pi^2 \hbar^2 k^2 / (2 m L^2)$ for all $k \in \Z$ and it is assumed that $\sum_{k \in \Z} \sum_{q \in \N_0} n_{k, q} = n$. Since we are investigating open systems in contact with an infinite particle reservoir, in thermal equilibrium the occupation numbers $n_{k, q}$ can be obtained from the grand canonical ensemble, and they are given by the Bose-Einstein (\enquote{$-$}), respectively, Fermi-Dirac (\enquote{$+$}) distribution \cite{huang,mazenko}:
\begin{align}\label{eq:occupationNumberDist}
    \begin{split}
        n_{k, q} &\equiv \braket{n_{k, q}} \\
        &= \frac{1}{\exp\Bigl\{\beta \bigl[\varepsilon_k + \hbar \omega \left(q + \frac{1}{2}\right) - \mu\bigr]\Bigr\} \mp 1} \ .
    \end{split}
\end{align}

\begin{remark}\label{rem:conditionQ}
    It is noteworthy to observe at this point that for bosons, one can actually see that the strict condition on the quantum number $q$, derived with the argument of positivity of the effective frequency in \cref{eq:conditionQ}, is recovered by requiring the occupation number $n_{k, q}$ to be non-negative, which means that the exponent in the above expression must be non-negative, see also \cref{app:idealBoseGas}.
\end{remark}

Assuming either Bose-Einstein or Fermi-Dirac statistics and using the above expression for $n_{k, q}$, one can show convergence of \eqref{eq:energyValuesTraOsc}.

\begin{lemma}\label{lem:convergenceTraOsc}
    The following double series converges: $\sum_{k \in \Z} \sum_{q \in \N_0} \left[\varepsilon_k + \hbar \omega \left(q + \frac{1}{2}\right)\right] n_{k, q}$.
\end{lemma}

\begin{proof}
    For simplicity, we set the multiplicative constants $4 \pi^2 \hbar^2 / (2 m L^2) \equiv 1$ and $\hbar \omega \equiv 1$ to unity, and we assume that $\beta = 1$ for simplicity (that is, we use $k_B T$ as energy units). Note that the constant term appearing in \eqref{eq:energyValuesTraOsc} evaluates to $\sum_{k \in \Z} \sum_{q \in \N_0} \frac{1}{2} \, n_{k, q} = \frac{1}{2} \, n$ by construction of the number operator and its eigenvalues, hence we may ignore it in the remainder of the proof as it does not have any influence on the convergence of the entire expression. Thus, we have to investigate the series
    \begin{equation*}
        S \ce \sum_{k \in \Z} \sum_{q \in \N_0} \frac{k^2 + q}{\ee^{k^2 + q} \, \ee^{1/2 - \mu} \mp 1} \ .
    \end{equation*}
    
    For a fixed $k \in \Z$, introduce the new variable $r \ce q + k^2$, which takes the values $\set{k^2, k^2 + 1, k^2 + 2, \dotsc} \subset \N_0$. With this, $S$ can be rewritten as
    \begin{align}\label{eq:auxiliarySeriesS}
        \begin{split}
            S &= \sum_{k \in \Z} \, \sum_{r=k^2}^{\infty} \frac{r}{\ee^{r} \, \ee^{1/2 - \mu} \mp 1} \\
            &= 2 \sum_{k=1}^{\infty} \sum_{r=k^2}^{\infty} \frac{r}{\ee^{r} \, \ee^{1/2 - \mu} \mp 1} + \sum_{r=0}^{\infty} \frac{r}{\ee^{r} \, \ee^{1/2 - \mu} \mp 1} \ .
        \end{split}
    \end{align}
    Recall that from the series representation $\ee^{\xi} = \sum_{n \in \N_0} \xi^n / n!$, $\xi \in \R$, of the exponential function, one obtains $\ee^\xi \ge \xi^{m + 1} / (m + 1)!$ for all $\xi \ge 0$ and $m \in \N_0$, hence we may use $\ee^r \ge r^5 / 5!$ to estimate \eqref{eq:auxiliarySeriesS}. Setting $C \ce \ee^{1/2 - \mu}$, it immediately follows that in the fermionic (\enquote{$+$}) case,
    \begin{equation*}
        C \, \ee^r + 1 \ge \frac{C}{5!} \, r^5 \ .
    \end{equation*}
    For the bosonic (\enquote{$-$}) case, observe that the chemical potential is always strictly smaller than the ground state energy, see \cref{eq:existenceBoseGas}; in our choice of units, this means that $\mu < \frac{1}{2}$, and hence $C > 1$. Therefore,
    \begin{equation*}
        C \, \ee^r - 1 = C \sum_{n=1}^{\infty} \frac{r^n}{n!} + (C - 1) \ge \frac{C}{5!} \, r^5 \ .
    \end{equation*}
    Thus, in total, we have the following estimate:
    \begin{equation*}
        C \, \ee^r \mp 1 \ge \frac{C}{5!} \, r^5 \ .
    \end{equation*}
    With this, it follows that
    \begin{equation*}
        S \le \frac{240}{C} \sum_{k=1}^{\infty} \sum_{r=k^2}^{\infty} \frac{1}{r^4} + \frac{120}{C} \sum_{r=1}^{\infty} \frac{1}{r^4} \ .
    \end{equation*}
    The second series is well-known to converge: $\frac{5!}{C} \sum_{r=1}^{\infty} \frac{1}{r^4} = \frac{5!}{C} \, \zeta(4) = \frac{5!}{C} \, \frac{\pi^4}{90} = \frac{4 \pi^4}{3 C}$. The first term can be rewritten as follows:
    \begin{equation*}
        \wt{S} \ce \sum_{k=1}^{\infty} \sum_{r=k^2}^{\infty} \frac{1}{r^4} = \sum_{k=1}^{\infty} \frac{1}{3!} \, \psi^{(3)}(k^2) \ ,
    \end{equation*}
    where we have introduced the polygamma function $\psi^{(3)}$ \cite{Abramowitz1964}. This function satisfies the following inequality \cite{Qi2010}:
    \begin{equation*}
        \psi^{(3)}(k^2) \le \frac{2!}{\bigl(k^2\bigr)^3} + \frac{3!}{\bigl(k^2\bigr)^4} = \frac{2}{k^6} + \frac{6}{k^8} \ .
    \end{equation*}
    Therefore, $\wt{S}$ is bounded as follows:
    \begin{align*}
        \wt{S} &\le \frac{1}{6} \sum_{k=1}^{\infty} \left[\frac{2}{k^6} + \frac{6}{k^8}\right] \\
        &= \frac{1}{6} \Bigl(2 \, \zeta(6) + 6 \, \zeta(8)\Bigr) \\
        &= \frac{1}{3} \frac{\pi^6}{945} + \frac{\pi^8}{9450} \ .
    \end{align*}
    In conclusion, the series $S$ is bounded above by converging series; hence, it is finite and converges by the series comparison test,
    \begin{align*}
        S &\le \frac{240}{C} \, \wt{S} + \frac{4 \pi^4}{3 C} \\
        &\le \frac{1}{C} \left(\frac{4 \pi^4}{3} + \frac{16 \pi^6}{189} + \frac{8 \pi^8}{315}\right) < + \infty \ .\tag*{\qedhere}
    \end{align*}
\end{proof}

\begin{lemma}\label{lem:convergenceOsc}
    The series $\sum_{q \in \N_0} \hbar \omega \left(q + \frac{1}{2}\right) n_q$, describing the eigenvalue of the Hamiltonian \eqref{eq:secondQuantizedHamiltonian} in the occupation number state $\Phi_n(\dotsc, n_q, \dotsc)$ with $n_q$ given by \eqref{eq:occupationNumberDist} for $\varepsilon_k \equiv 0$, converges.
\end{lemma}

\begin{proof}
    The argument is essentially that provided in the previous proof for the convergence of the second series appearing in \cref{eq:auxiliarySeriesS}.
\end{proof}

\begin{corollary}\label{cor:effectiveEnergies}
    The series appearing in the effective energies \eqref{eq:energySpectrumVibrations} and \eqref{eq:energySpectrumTranslationsVibrations} converge; hence, they may be combined to a single series, that is, we may write
    \begin{equation*}
        E_n^\mathrm{eff} = \sum_{q=0}^{\infty} \left[\hbar \omega \left(q + \frac{1}{2}\right) - \mu\right] n_{q} \ ,
    \end{equation*}
    respectively,
    \begin{equation*}
        E_n^\mathrm{eff} = \sum_{k \in \Z} \sum_{q \in \N_0} \left[\varepsilon_k + \hbar \omega \left(q + \frac{1}{2}\right) - \mu\right] n_{k, q} \ .
    \end{equation*}
\end{corollary}

\begin{proof}
    Use \cref{lem:convergenceOsc}, respectively, \cref{lem:convergenceTraOsc}, together with the fact that the series $\sum_{q \in \N_0} n_q$, respectively, $\sum_{k \in \Z} \sum_{q \in \N_0} n_{k, q}$, is finite (and evaluates to the particular number $n$), the two converging series may be added together term by term.
\end{proof}

\bibliography{harm.bib}

\end{document}